\documentclass[11pt]{article}

\usepackage{fullpage,times,latexsym,amsmath,amsfonts}

\def\01{\{0,1\}}
\newcommand{\ceil}[1]{\lceil{#1}\rceil}

\newcommand{\eps}{\varepsilon}

\newcommand{\ket}[1]{|#1\rangle}

\newcommand{\ketbra}[2]{|#1\rangle\langle#2|}

\newtheorem{fact}{Fact}
\newtheorem{claim}{Claim}
\newtheorem{theorem}{Theorem}

\newtheorem{corollary}{Corollary}

\newenvironment{proof}
{\noindent {\bf Proof. }}
{{\hfill $\Box$}\\
\smallskip}

\begin{document}

\title{Optimizing the Number of Gates in Quantum Search}
\author{Srinivasan Arunachalam\thanks{QuSoft, CWI Amsterdam, the Netherlands.
Supported by ERC Consolidator Grant QPROGRESS.} 
\and 
Ronald de Wolf\thanks{QuSoft, CWI and University of Amsterdam, the Netherlands.
Partially supported by ERC Consolidator Grant QPROGRESS and by the European Commission FET-Proactive project Quantum Algorithms (QALGO) 600700.}
}
\maketitle

\begin{abstract}
In its usual form, Grover's quantum search algorithm uses $O(\sqrt{N})$ queries and $O(\sqrt{N} \log N)$ other elementary gates to find a solution in an $N$-bit database.  Grover in 2002 showed how to reduce the number of other gates to $O(\sqrt{N}\log\log N)$ for the special case where the database has a unique solution, without significantly increasing the number of queries.  We show how to reduce this further to $O(\sqrt{N}\log^{(r)} N)$ gates for any constant~$r$, and sufficiently large~$N$.  This means that, on average, the circuits between two queries barely touch more than a constant number of the $\log N$ qubits on which the algorithm acts. For a very large $N$ that is a power of~2, we can choose~$r$ such that the algorithm uses essentially the minimal number $\frac{\pi}{4}\sqrt{N}$ of queries, and only $O(\sqrt{N}\log(\log^{\star} N))$ other gates.
\end{abstract}

\section{Introduction}

One of the main successes of quantum algorithms so far is Grover's algorithm for \emph{database search}~\cite{grover:search,bhmt:countingj}. Here a database of size~$N$ is modeled as a binary string $x\in\01^N$, whose bits are indexed by $i\in\{0,\ldots,N-1\}$. A \emph{solution} is an index~$i$ such that $x_i=1$. The goal of the search problem is to find such a solution given access to the bits of $x$. If our database has Hamming weight $|x|=1$, we say it has a \emph{unique}~solution.

The standard version of Grover's algorithm finds a solution with high probability using $O(\sqrt{N})$ database queries and $O(\sqrt{N}\log N)$ other elementary gates. It starts from a uniform superposition over all database-indices~$i$, and then applies $O(\sqrt{N})$ identical ``iterations,'' each of which uses one query and $O(\log N)$ other elementary gates. Together these iterations concentrate most of the amplitude on the solution(s). A measurement of the final state then yields a solution with high probability. For the special case of a database with a unique solution its number of iterations ($=$ number of queries) is essentially $\frac{\pi}{4}\sqrt{N}$, and Zalka~\cite{zalka:grover} showed that this number of queries is optimal. Grover's algorithm, in various forms and generalizations, has been applied as a subroutine in many other quantum algorithms, and is often the main source of speed-up for those.  See for example~\cite{bht:collision,BuhrmanCleveWigderson98,betal:distinctnessj,durr&hoyer:minimum,dhhm:graphproblems,dorn:thesis}.

In~\cite{grover:tradeoffs}, Grover gave an alternative algorithm to find a unique solution using slightly more (but still $(\frac{\pi}{4}+o(1))\sqrt{N}$) queries, and only $O(\sqrt{N}\log\log N)$ other elementary gates. The algorithm is more complicated than the standard Grover algorithm, and no longer consists of $O(\sqrt{N})$ identical iterations. Still, it acts on $O(\log N)$ qubits, so on average a unitary sitting between two queries acts on only a tiny $O(\log\log N / \log N)$ fraction of the qubits.  It is quite surprising that such mostly-very-sparse unitaries suffice for quantum search.

In this paper we show how Grover's reduction in the number of gates can be improved further: for every fixed~$r$, and sufficiently large $N$, we give a quantum algorithm that finds a unique solution in a database of size~$N$ using $O(\sqrt{N})$ queries and $O(\sqrt{N}\log^{(r)}N)$ other elementary gates.\footnote{The constant in the $O(\cdot)$ depends on $r$.  The iterated binary logarithm is defined as $\log^{(s+1)}= \log \circ \log ^{(s)}$, where $\log^{(0)}$ is the identity function. The function $\log^{\star} N$ is the number of times the binary logarithm must be iteratively applied to~$N$ to obtain a number that is at most~1: $\log^{\star} N=\min \{r \geq 0 : \log^{(r)} N\leq 1\}$.}
To be concrete about the latter, we assume the set of elementary gates at our disposal is the Toffoli gate and all one-qubit unitary gates (including the Hadamard gate~$H$ and the Pauli $X$ gate).

Our approach is recursive: we build a quantum search algorithm for a larger database using amplitude amplification on a search algorithm for a smaller database.%
\footnote{The idea of doing recursive applications of amplitude amplification to search increasingly larger database-sizes is reminiscent of the algorithm of Aaronson and Ambainis~\cite{aaronson&ambainis:searchj} for searching an $N$-element database that is arranged in a $d$-dimensional grid.  However, their goal was to find a local search algorithm with optimal number of \emph{queries} (they succeeded for $d>2$), not to optimize the number of \emph{gates}. If one writes out their algorithm as a quantum circuit, it still has $O(\sqrt{N}\log N)$ gates.}
Let us sketch this in a bit more detail.  Suppose we have a sequence of database-sizes $N_1,\ldots,N_r=N$, where $N_{i+1}\approx 2^{\sqrt{N_i}}$ (of course, $N$ needs to be sufficiently large for such a sequence to exist). The basic Grover algorithm can search a database of size~$N_1$~using 
$$
Q_1=O(\sqrt{N_1}),~~~~E_1=O(\sqrt{N_1}\log N_1)
$$ 
queries and gates, respectively. We can build a search algorithm for database-size $N_2$ as follows. Think of the $N_2$-sized database as consisting of $N_2/N_1$ $N_1$-sized databases; we can just pick one such $N_1$-sized database at random, use the smaller algorithm to search for a solution in that database, and then use $O(\sqrt{N_2/N_1})$ rounds of amplitude amplification to boost the $N_1/N_2$ probability that our randomly chosen $N_1$-sized database happened to contain the unique solution. Each round of amplitude amplification involves one application of the smaller algorithm, one application of its inverse, a reflection through the $\log N_2$-qubit all-0 state, and one more query. This gives a search algorithm for an $N_2$-sized database that uses
$$
Q_{2}=O\left(\sqrt{\frac{N_{2}}{N_1}}Q_1\right)=O(\sqrt{N_2}),~~~~E_{2}=O\left(\sqrt{\frac{N_{2}}{N_1}}(E_1+\log N_2)\right)
$$ 
queries and gates respectively.  Note that by our choice of $N_2\approx 2^{\sqrt{N_1}}$, we have $E_1\geq \sqrt{N_1}$, so $E_{2}=O(\sqrt{N_{2}/N_1}E_1)$. Repeating this construction gives a recursion
$$
Q_{i+1}=O\left(\sqrt{\frac{N_{i+1}}{N_i}}Q_i\right),~~~~E_{i+1}=O\left(\sqrt{\frac{N_{i+1}}{N_i}}E_i\right).
$$ 
The constant factor in the $O(\cdot)$ blows up by a constant in each recursion, so after $r$ steps this unfolds to
$$
Q_{r}=O(\exp(r)\sqrt{N}),~~~~E_{r}=O(\exp(r)\sqrt{N}\log N_1).
$$ 
Here $\log N_1=O(\log^{(r)} N)$ because $N_1,\ldots,N_r=N$ is (essentially) an exponentially increasing sequence.

The result we prove in this paper is stronger: it does not have the $\exp(r)$ factor. Tweaking the above idea to avoid this factor is somewhat delicate, and will take up the remainder of this paper. For instance, in order to get close to the optimal query complexity $\frac{\pi}{4}\sqrt{N}$, it is important that the intermediate steps do not amplify the success probability all the way to~1, since amplitude amplification is less efficient when boosting large success probabilities to~1 than when boosting small success probabilities to somewhat larger success probabilities. Our final algorithm will boost the success probability to~1 only at the very end, after all $r$ recursion steps have been done.

Choosing $r=\log^\star N$ in our result and being careful about the constants, we get an exact quantum algorithm for finding a unique solution using essentially the optimal $\frac{\pi}{4}\sqrt{N}$ queries and $O(\sqrt{N}\log(\log^\star N))$ elementary gates. Note that for the latter algorithm, on average there are only $O(\log(\log^\star N))$ elementary gates in between two queries, which is barely more than constant. Once in a while a unitary acts on many more qubits, but the average is only $O(\log(\log^\star N))$.

\medskip

{\bf Possible objections.} To pre-empt the critical reader, let us mention two objections one may raise against the fine-grained optimization of the number of elementary gates that we do here.  First, one query acts on $\log N$ qubits, and when itself implemented using elementary gates, any oracle that's worth its salt would require $\Omega(\log N)$ gates. Since $\Omega(\sqrt{N})$ queries are necessary, a fair way of counting would say that just the queries themselves already have ``cost'' $\Omega(\sqrt{N}\log N)$, rendering our (and Grover's~\cite{grover:tradeoffs}) gate-optimizations moot.
Second, to do exact amplitude amplification in our recursion steps, we allow infinite-precision single-qubit phase gates.  This is not realistic, as in practice such gates would have to be approximated by more basic gates.
Our reply to both would be: fair enough, but we still find it quite surprising that query-efficient search algorithms only need to act on a near-constant number of qubits in between the queries on average. It is interesting that after nearly two decades of research on quantum search, the basic search algorithm can still be improved in some ways. It may even be possible to optimize our results further to use $O(\sqrt{N})$ elementary gates, which would be even more surprising.

\section{Preliminaries}

Let $[n]=\{1,\ldots,n\}$. We use the binary logarithm throughout this paper. We will typically assume for simplicity that the database-size $N$ is a power of~2, $N=2^n$, so we can identify indices $i$ with their binary representation $i_1\ldots i_n\in\01^n$. We can access the database by means of \emph{queries}.  A query is the following unitary map on $n+1$ qubits:
$$
O_x:\ket{i,b}\mapsto\ket{i,b\oplus x_i},
$$
where $i\in\{0,\ldots,N-1\}$ and $b\in\01$. Given access to an oracle of the above type, we can also make a phase query of the form $O_{x,\pm}:\ket{i}\rightarrow (-1)^{x_i}\ket{i}$ by the standard ``phase kickback trick.'' 

Let $D_n=2\ketbra{0^n}{0^n}-I$ be the $n$-qubit unitary that reflects through $\ket{0^n}$. It is not hard to see that this can be implemented using $O(n)$ elementary gates and $n-1$ ancilla qubits that all start and end in $\ket{0}$ (and that we often will not even write explicitly). Specifically, one can use $X$ gates to each of the $n$ qubits, then use $n-1$ Toffoli gates into $n-1$ ancilla qubits to compute the logical AND of the first $n$ qubits, then apply $-Z$ to the last qubit (which negates the basis states where this AND is~0), and reverse the Toffolis and $X$s.

Amplitude amplification is a technique that can be used to efficiently boost quantum search algorithms with a known success probability~$a$ to higher success probability.  We will invoke the following theorem from~\cite{bhmt:countingj} in the proof of Theorem~\ref{thm:primitivealgo} later. For the sake of completeness we include a proof in the~appendix.

\begin{theorem}
\label{thm:exactamplitude}
Let $N=2^n$.
Suppose there exists a unitary quantum algorithm $\mathcal{A}$ that finds a solution in database $x\in\01^N$ with known probability~$a$, in the sense that measuring $\mathcal{A}\ket{0^n}$ yields a solution with probability exactly~$a$. Let $a<a'\in [0,1]$ and $w=\ceil{\frac{\arcsin(\sqrt{a'})}{2\arcsin(\sqrt{a})}-\frac{1}{2}}$. Then there exists a quantum algorithm $\mathcal{B}$ that finds a solution with probability exactly $a'$ using $w+1$ applications of algorithm $\mathcal{A}$, $w$ applications of $\mathcal{A}^{-1}$, $w$ additional queries, and $4w(n+2)$ additional elementary gates. In total, $\mathcal{B}$ uses $(2w+1)Q+w$ queries and $w(4n+2E+8)+E$ elementary gates.
\end{theorem}

Note that if $\mathcal{A}=H^{\otimes n}$ and our $N$-bit database has a unique solution, then $a=1/N$. For $k\geq 2$ and $a'=1/k$, Theorem~\ref{thm:exactamplitude} implies an algorithm $\mathcal{C}^{(1)}$ that finds a solution with probability exactly $1/k$ using $w$ queries and at most $O(w\log N)$ other elementary gates, where $w\leq\ceil{\frac{\sqrt{N}(1+1/k)}{2\sqrt{k}}-\frac{1}{2}}$ (this upper bound on~$w$ follows because $\arcsin(z)\geq z$, and $\sin(\frac{1+1/k}{\sqrt{k}})\geq \frac{1}{\sqrt{k}}$ since $\sin(z)\geq z-z^3/6$ for $z\geq 0$). 

In order to amplify the probability of an algorithm from $1/k$ to $1$ we use the following corollary.

\begin{corollary}
\label{cor:amplify1/kto1}
Let $k\geq 2$, $n$ be integers, $N=2^n$. Suppose there exists a quantum algorithm $\mathcal{D}$ that finds a unique solution in an $N$-bit database with probability exactly $1/k$ using $Q\geq \sqrt{k}$ queries and $E$ elementary gates. Then there exists a quantum algorithm that finds the unique solution with probability~1 using at most $\frac{\pi}{2}Q\sqrt{k}(1+\frac{2}{\sqrt{k}})^2$ queries and $O(\sqrt{k}(n+E))$ elementary gates.
\end{corollary}

\begin{proof}
Applying Theorem~\ref{thm:exactamplitude} to algorithm $\mathcal{D}$ with $a=1/k, a'=1$, we obtain an algorithm that succeeds with probability~1 using at most $w'(2Q+1)+Q$ queries and $O(w'(n+E))$ gates. Here
\begin{align*}
\begin{aligned}
w'&=\Big\lceil \frac{\arcsin(1)}{2\arcsin(1/\sqrt{k})}-\frac{1}{2}\Big\rceil\leq \frac{\pi}{4}(\sqrt{k}+1),
\end{aligned}
\end{align*}
using $\arcsin(x)\geq x$ and $\ceil{\frac{\pi}{4}\sqrt{k}-\frac{1}{2}} \leq \frac{\pi}{4}(\sqrt{k}+1)$. Hence, the total number of queries in this new algorithm is at most 
\begin{align*}
\frac{\pi}{4}(\sqrt{k}+1)(2Q+1)+ Q &= \frac{\pi}{2}Q(\sqrt{k}+1)\Big(1+\frac{1}{2Q}+\frac{2}{\pi(\sqrt{k}+1)}\Big) \\
&\leq \frac{\pi}{2}Q(\sqrt{k}+1)(1+\frac{2}{\sqrt{k}})\\
&\leq \frac{\pi}{2}Q\sqrt{k}(1+\frac{2}{\sqrt{k}})^2,
\end{align*}
where we used $Q\geq \sqrt{k}$ and $2\sqrt{k}\leq \pi(\sqrt{k}+1)$ in the first inequality. The total number of gates is $O(\sqrt{k}(n+E))$.
\end{proof}

The following easy facts will be helpful to get rid of some of the ceilings that come from Theorem~\ref{thm:exactamplitude}.

\begin{fact}
\label{fact:ceiltononceil}
If $k\geq 2$ and $\alpha\geq k$, then $\ceil{\frac{\alpha}{2}(1+\frac{1}{k})-\frac{1}{2}}\leq \frac{\alpha}{2}(1+\frac{2}{k})$.
\end{fact}

\begin{fact}
\label{fact:logupperbound}
If $k\geq 3$ and $i\geq 2$, then $(2i+8)\log k< k^{i+1}$.
\end{fact}

\begin{proof}
Fixing $i=2$, it is easy to see that $12\log k< k^3$ for $k\geq 3$. Similarly, fix $k=3$ and observe that $(2i+8)\log 3< 3^{i+1}$ for all $i\geq 2$. This implies the result for all $k\geq 3$ and $i\geq 2$, because the right-hand side grows faster than the left-hand side in both $i$ and~$k$.
\end{proof}

\section{Improving the gate complexity for quantum search}
\label{sec:gatecomplexity}

In this section we give our main result, which will be proved by recursively applying the following theorem.

\begin{theorem}
\label{thm:primitivealgo}
Let $k\geq 4$,  $n\geq m+2\log k$ be integers, $M=2^m$ and $N=2^n$.
Suppose there exists a quantum algorithm $\mathcal{G}$ that finds a unique solution in an $M$-bit database with a known success probability that is at least $1/k$, using $Q\geq k+2$ queries and $E$ other elementary gates.  Then there exists a quantum algorithm that finds a unique solution in an $N$-bit database with probability exactly $1/k$, using $Q'$ queries and $E'$  other elementary gates where,
\begin{align*}
&Q'\leq Q\sqrt{N/M}(1+4/k), \qquad E\sqrt{N/M} \leq  E' \leq (3n+E)\sqrt{N/M}(1+3/k).
\end{align*}
\end{theorem}

\begin{proof}
Consider the following algorithm $\mathcal{A}$:
\begin{enumerate}
\item Start with $\ket{0^n}$.
\item Apply a Hadamard transform to the first $n-m$ qubits, leaving the last $m$ qubits as $\ket{0^m}$. The resulting state is a uniform superposition over the first $n-m$ qubits $\frac{1}{\sqrt{N/M}}\sum_{y\in \01^{n-m}}\ket{y}\ket{0^m}$.
\item Apply the unitary $\mathcal{G}$ to the last $m$ qubits (using queries to $x$, with the first $n-m$ address bits fixed).
\end{enumerate}
The final state of algorithm $\mathcal{A}$ is
$$
(H^{\otimes(n-m)} \otimes \mathcal{G}) \ket{0^n}=\frac{1}{\sqrt{N/M}}\sum_{y\in\01^{n-m}}\ket{y}\mathcal{G}\ket{0^m}.
$$
The state $\ket{y}\mathcal{G}\ket{0^m}$ depends on $y$, because here $\mathcal{G}$ restricts to the $M$-bit database that corresponds to the bits in~$x$ whose address starts with~$y$.  
Let $t$ be the $n$-bit address corresponding to the unique solution in the database $x\in\01^N$. Then the probability of observing $\ket{t_1\ldots t_n}$ in the state $\ket{t_1\ldots t_{n-m}}\mathcal{G}\ket{0^m}$ is at least~$1/k$. Hence the probability that $\mathcal{A}$ finds the solution is $a \geq \frac{M}{kN}$. The total number of queries of algorithm $\mathcal{A}$ is $Q$ (from Step 3) and the total number of elementary gates is $n-m+E$ (from Steps 2 and~3). 

Applying Theorem~\ref{thm:exactamplitude} to algorithm $\mathcal{A}$ by choosing $a'=1/k$, we obtain an algorithm $\mathcal{B}$ using at most $w(2Q+1)+Q$ queries and  $w(4n+2E+8)+E$ gates (from Theorem~\ref{thm:exactamplitude}), where
\begin{equation*}
\begin{aligned}
w&=\Big\lceil\frac{\arcsin(\sqrt{a'})}{2\arcsin(\sqrt{a})}-\frac{1}{2}\Big\rceil\leq \Big\lceil\frac{\sqrt{1/k}(1+1/k)}{2\sqrt{a}}-\frac{1}{2}\Big\rceil 
\leq \Big\lceil\frac{\sqrt{N}(1+1/k)}{2\sqrt{M}}-\frac{1}{2}\Big\rceil\leq \frac{\sqrt{N}(1+2/k)}{2\sqrt{M}},\\
\end{aligned}
\end{equation*}
where the first inequality follows from $\arcsin(z)\geq z$ and $\sin(\frac{1+1/k}{\sqrt{k}})\geq \frac{1}{\sqrt{k}}$ (since $\sin(z)\geq z-z^3/6$ for $z\geq 0$), and the third inequality uses Fact~\ref{fact:ceiltononceil} ($\sqrt{N/M}\geq k$ because $n\geq m+2\log k$). 

The total number of queries in algorithm $\mathcal{B}$ is at most
\begin{align*}
w(2Q+1)+Q &\leq Q\sqrt{N/M} (1+2/k)+\frac{1}{2}\sqrt{N/M} (1+2/k)+Q \\
&\leq Q\sqrt{N/M}(1+2/k)+\frac{Q}{2k}\sqrt{N/M}+\frac{Q}{k}\sqrt{N/M}\\
&\leq Q\sqrt{N/M}(1+4/k)
\end{align*}
where we used $Q\geq k+2$ and $n\geq m+2\log k\geq 4$ in the second inequality. The number of gates in $\mathcal{B}$ is
\begin{align*}
w(4n+2E+8)+E &\leq \sqrt{N/M} (1+2/k)(2n+E+4)+E\leq (3n+E)\sqrt{N/M} (1+3/k),
\end{align*}
where we used $n\geq m+2\log k\geq 4$ (since $k\geq 4$) and $E\leq \frac{E}{k}\sqrt{N/M}$ in the second inequality. 

It is not hard to see that the number of gates in $\mathcal{B}$ is at least~$E\sqrt{N/M}$.
\end{proof}

Applying Theorem~\ref{thm:exactamplitude} once to an algorithm that finds the unique solution in an $M$-bit database with probability~$1/\log\log N$, we get the following corollary, which was essentially the main result of Grover~\cite{grover:tradeoffs}.

\begin{corollary}
\label{cor:grovertradeoff}
Let $n\geq 25$ and $N=2^n$.
There exists a quantum algorithm that finds a unique solution in a database of size $N$ with probability~1, using at most $(\frac{\pi}{4}+o(1))\sqrt{N}$ queries and $O(\sqrt{N} \log \log N)$ other elementary gates.
\end{corollary}

\begin{proof}
Let $m=\ceil{\log(n^2k^3)}$ and $k= \log \log N$.
Let $\mathcal{C}^{(1)}$ be the algorithm (described after Theorem~\ref{thm:exactamplitude}) on an $M$-bit database with $M = 2^m$ that finds the solution with probability~$1/k$. Observe that $k\geq 4$ and $m+2\log k\leq \log (2n^2k^5)\leq n$ (where the last inequality is true for $n\geq 25$), hence we can apply Theorem~\ref{thm:primitivealgo} using $\mathcal{C}^{(1)}$ as our base algorithm. This gives an algorithm $\mathcal{C}^{(2)}$ that finds the solution with probability exactly $1/k$. The total number of queries in algorithm $\mathcal{C}^{(2)}$ is at most
\begin{equation*}
\begin{aligned}
\Big\lceil\frac{\sqrt{M}(1+1/k)}{2\sqrt{k}}-\frac{1}{2}\Big\rceil\cdot \Big(\sqrt{N/M}(1+4/k)\Big) &\leq \frac{\sqrt{M}(1+2/k)}{2\sqrt{k}}\sqrt{N/M}(1+4/k) \\
&\leq \sqrt{\frac{N}{4k}}(1+4/k)^2,
\end{aligned}
\end{equation*}
where  the expression on the left is the contribution from Theorem~\ref{thm:primitivealgo}. The first inequality above follows from Fact~\ref{fact:ceiltononceil} (since $m\geq 4\log k$). The total number of gates in $\mathcal{C}^{(2)}$ is 
\begin{equation*}
\begin{aligned}
O\Big(\Big(3n+\Big\lceil\frac{\sqrt{M}(1+\frac{1}{k})}{2\sqrt{k}}-\frac{1}{2}\Big\rceil\log M\Big)\sqrt{\frac{N}{M}}(1+\frac{3}{k})\Big) &\leq O \Big(\sqrt{\frac{N}{k}}\Big(\frac{3n\sqrt{k}(1+3/k)}{\sqrt{M}}+(1+3/k)^2 \log M\Big)\Big)\\
& \leq O\Big(\sqrt{\frac{N}{k}}\Big(1+\frac{3}{k}\Big)^{3}\log\log N\Big),
\end{aligned}
\end{equation*}
where we used Fact~\ref{fact:ceiltononceil} in the first inequality, $n\sqrt{k}(1+1/k)\leq \sqrt{M}/k$ (since $m\geq\log(n^2k^3)$) and~$\log M=O(\log \log N)$ in the second inequality. Applying Corollary~\ref{cor:amplify1/kto1} to algorithm $\mathcal{C}^{(2)}$, we obtain an algorithm that succeeds with probability~1 using at most 
$$
\frac{\pi}{2}\Big(\sqrt{\frac{N}{4k}}(1+\frac{4}{k})^2\Big)\cdot \Big(\sqrt{k}(1+\frac{2}{\sqrt{k}})^2\Big) \leq \frac{\pi}{4} \sqrt{N}\Big(1+\frac{4}{\sqrt{k}}\Big)^4
$$
queries and
$$
O\Big(n\sqrt{k}+\sqrt{N}\Big(1+\frac{3}{k}\Big)^{3}\log \log N\Big) \leq O\Big(\sqrt{N}\Big(1+\frac{3}{k}\Big)^{3}\log\log N\Big) 
$$ 
gates, since $n\sqrt{k}\leq \sqrt{N}\log\log N$ (which is true for $n\geq 25$). Since $k= \log \log N$, it follows that the query complexity is  at most $\frac{\pi}{4}\sqrt{N} (1+o(1))$ and the gate complexity is~$O(\sqrt{N}\log\log N)$.
\end{proof}

We can now use Theorem~\ref{thm:primitivealgo} recursively by starting from the improved algorithm from Corollary~\ref{cor:grovertradeoff}. This gives query complexity $O(\sqrt{N})$ and gate complexity $O(\sqrt{N} \log\log\log N)$. Doing this multiple times and being careful about the constant (which grows in each step of the recursion), we obtain the following result:

\begin{theorem}
\label{thm:grover-almost-opt}
Let $k$ be a power of~2 and $N$ a sufficiently large power of~2. For every $r\in [\log^\star N]$, $k\in \{4,\ldots, \log\log N\}$, there exists a quantum algorithm that finds a unique solution in a database of size $N$ with probability exactly $1/k$, using at most 
$$
\sqrt{\frac{N}{4k}}(1+4/k)^{r}\mbox{ queries and }O\left(\sqrt{\frac{N}{k}}(1+6/k)^{2r-1} \max\{\log k,\log^{(r)}N\}\right)\mbox{ other elementary~gates}.
$$ 
\end{theorem}

\begin{proof}
We begin by defining a sequence of integers $n_1,\ldots,n_r$:
$$
n_r=\log N\mbox{ and }n_{i-1}=\max \{(2i+6)\log k,\ceil{\log (n^2_ik^3)}\},\mbox{ for }i\in\{2, \ldots, r\}.
$$
Note that $n_1\geq 10\log k\geq 20$ since $k\geq 4$. We first prove the following claim about this sequence.

\begin{claim}
\label{claim:increasingdatabases}
If $i\in \{2,\ldots,r\}$, then $n_{i-1}+2\log k\leq n_i$ .
\end{claim}

\begin{proof} 
We use downward induction on $i$. For the base case $i=r$, note that $n_r=\log N$. Since $n_{r-1}=\max \{(2r+6)\log k,\ceil{\log (n^2_rk^3)}\}$, and $(2r+6)\log k\leq \ceil{\log (n^2_rk^3)}$ for sufficiently large $N$ and $k\leq \log\log N$, we may assume $n_{r-1}=\ceil{\log (n^2_rk^3)}$. Hence
$$
n_{r-1}+2\log k=\ceil{\log (n^2_rk^3)}+2\log k\leq \log (2n^2_rk^5) \leq \log N=n_r,
$$
where the last inequality assumed $N$ sufficiently large and used $k\leq\log\log N$.

For the inductive step, assume we have $n_j+2\log k\leq n_{j+1}$. We now prove $n_{j-1}+2\log k\leq n_j$ by considering the two possible values for $n_{j-1}$.

{\bf Case~1.} $n_{j-1}=(2j+6)\log k$. Then we have 
\begin{equation}
\label{eq:increasingdatabases}
n_{j-1}+2\log k \begin{cases} 
 =n_j & \text{if }  n_{j}=(2j+8)\log k \\
 \leq \ceil{\log (n^2_{j+1}k^3)}=n_j & \text{if }  n_{j}=\ceil{\log (n^2_{j+1}k^3)},
\end{cases}
\end{equation}
using $n_{j}=\max \{(2j+8)\log k,\ceil{\log (n^2_{j+1}k^3)}\}$ in the last inequality. 

{\bf Case~2.} $n_{j-1}=\ceil{\log (n^2_jk^3)}$. We first show $n_{j-1}\leq n_j$:
\[n_{j-1}\leq\ceil{\log (n^2_{j+1}k^3)} \begin{cases} 
\leq (2j+8)\log k= n_j & \text{if } n_{j}=(2j+8)\log k \\
 =n_j & \text{if }  n_{j}=\ceil{\log (n^2_{j+1}k^3)},
\end{cases}
\]
where the first inequality uses induction hypothesis, the second uses $n_{j}=\max \{(2j+8)\log k,\ceil{\log (n^2_{j+1}k^3)}\}$. 
We can now conclude the inductive step:
$$
n_{j-1}+2\log k\leq  \log (2n^2_{j} k^5) = (1+2\log n_j)+5\log k\leq  n_j/2+5\log k\leq n_j/2+n_j/2 =n_j.
$$
In the first inequality above we used $n_{j-1}\leq \log(2n^2_jk^3)$, and we used $n_j\geq n_1\geq 10\log k\geq 20$ (using $n_{j-1}\leq n_j$ for $j\in \{2,\ldots ,r\}$ and $k\geq 4$) to conclude $1+2\log n_j\leq n_j/2$ (which is true for $n_j\geq 20$) in the second inequality, and $5\log k\leq n_j/2$ in the last inequality.
\end{proof}

Using the sequence $n_1,\ldots ,n_r$, we consider $r$ database-sizes $2^{n_1}=N_1\leq 2^{n_2}=N_2\leq \cdots \leq 2^{n_r}=N_r=N$. For each $i\in [r]$, we will construct a quantum algorithm $\mathcal{C}^{(i)}$ on a database of size $N_i$ that finds a unique solution with probability exactly $1/k$. $Q_i$ and $E_i$ will be the query complexity and gate complexity, respectively, of algorithm $\mathcal{C}^{(i)}$.  We have already constructed the required algorithm $\mathcal{C}^{(1)}$ (described after Theorem~\ref{thm:exactamplitude}) on an $N_1$-bit database~using
$$
Q_1= \Big\lceil\frac{\sqrt{N_1}(1+1/k)}{2\sqrt{k}}-\frac{1}{2}\Big\rceil \leq \frac{\sqrt{N_1}(1+2/k)}{2\sqrt{k}}
$$ 
queries, where the inequality follows from Fact~\ref{fact:ceiltononceil} (since $N_1\geq k^{10}$). Also, note that 
$$
Q_1\geq \frac{\sqrt{N_1}(1+1/k)}{2\sqrt{k}}-1 \geq k+2,
$$
where the first inequality used $N_1\geq k^{10}$, and the second inequality used $k\geq 4$. Using Theorem~\ref{thm:exactamplitude}, the number of gates $E_1$ used by $\mathcal{C}^{(1)}$ is
\begin{align*}
\Big\lceil\frac{\sqrt{N_1}(1+1/k)}{2\sqrt{k}}-\frac{1}{2}\Big\rceil(6\log N_1+8)+\log N_1 &\leq \frac{\sqrt{N_1}(1+2/k)}{\sqrt{k}}(3\log N_1+4)+\log N_1\\
& \leq \frac{4\sqrt{N_1}(1+2/k)}{\sqrt{k}}\log N_1+\log N_1\\
&\leq  \frac{4\sqrt{N_1}(1+3/k)}{\sqrt{k}}\log N_1,
\end{align*}
where we use Fact~\ref{fact:ceiltononceil} (since $N_1\geq k^{10}$) in the first inequality and $N_1\geq k^{10}$ in the second and third inequality. It is not hard to see that $E_1\geq \sqrt{N_1/(4k)}$. 

For $i\in\{2, \ldots, r\}$, we apply Theorem~\ref{thm:primitivealgo} using $\mathcal{C}^{(i-1)}$ as the base algorithm and we obtain an algorithm $\mathcal{C}^{(i)}$ that succeeds with probability exactly $1/k$. We showed earlier in Claim~\ref{claim:increasingdatabases} that $n_{i-1}+2\log k\leq n_i$ and it also follows that $k+2\leq Q_1\leq \cdots \leq Q_r$ (since the database-sizes $N_1,\ldots ,N_r$ are non-decreasing). Hence both assumptions of Theorem~\ref{thm:primitivealgo} are satisfied. The total number of queries used by $\mathcal{C}^{(i)}$~is
\begin{equation}
\label{eq:recursion1}
Q_i\leq \sqrt{\frac{N_{i}}{N_{i-1}}}Q_{i-1} \Big(1+\frac{4}{k}\Big).
\end{equation}
We need the following claim to analyze the number of gates used by $\mathcal{C}^{(i)}$: 

\begin{claim}
\label{claim:Eilowerbound}
$E_i\geq \sqrt{N_{i}/(4k)}$ for all $i\in [r]$.
\end{claim}

\begin{proof}
The proof is by induction on $i$. For the base case, we observed earlier that $E_1\geq \sqrt{N_1/(4k)}$. For the induction step assume $E_{i-1}\geq \sqrt{N_{i-1}/(4k)}$. The claim follows immediately from the lower bound on $E'$ in Theorem~\ref{thm:primitivealgo} since $E_i\geq E_{i-1}\sqrt{N_i/N_{i-1}}\geq \sqrt{N_i/(4k)}$.
\end{proof}

Recursively it follows that the number of gates $E_i$ used by $\mathcal{C}^{(i)}$ is at most 
\begin{equation}
\label{eq:recursion2}
\begin{aligned}
\sqrt{\frac{N_{i}}{N_{i-1}}(}E_{i-1}+3n_i) (1+3/k)&\leq \sqrt{\frac{N_{i}}{N_{i-1}}}E_{i-1}\Big(1+3n_i\sqrt{\frac{4k}{N_{i-1}}}\Big) (1+3/k)\\
& \leq \sqrt{\frac{N_{i}}{N_{i-1}}}E_{i-1} (1+6/k)^2,
\end{aligned}
\end{equation}
where we used Claim~\ref{claim:Eilowerbound} in the first inequality and $n_i\leq \sqrt{\frac{N_{i-1}}{k^3}}$ in the last inequality (note that this inequality also holds if $n_{i-1}=(2i+6)\log k\geq \ceil{\log (n^2_ik^3)}$). Unfolding the recursion in Equations~(\ref{eq:recursion1}) and~(\ref{eq:recursion2}), we~obtain
$$
Q_r\leq \sqrt{\frac{N_r}{4k}}\Big(1+\frac{4}{k}\Big)^{r}, \qquad E_r\leq 4\sqrt{\frac{N_r}{k}}\Big(1+\frac{6}{k}\Big)^{2r-1}\log N_1.
$$
It remains to show that $n_1$, which is defined to be $\max\{10\log k, \ceil{\log (n^2_2k^3)}\}$, is $O(\max\{\log k,\log^{(r)}N\})$. If $n_1=10\log k$, then we are done. If $n_1=\ceil{\log (n^2_2k^3)}$, we need the following claim to conclude the~proof.

\begin{claim}
\label{claim:logN1calculation}
Suppose $n_{1}=\ceil{\log(n^2_2k^3)}$. Then $n_{i-1}=\ceil{\log(n^2_ik^3)}$ for all $i\in \{2,\ldots,r\}$.
\end{claim}

\begin{proof}
We prove the claim by induction on $i$. The base case $i=2$ is the assumption of the claim. 

For the inductive step, assume $n_{i-1}=\ceil{\log(n^2_ik^3)}$ for some $i\geq 2$. We have 
$$
\log(n_i^2 k^4)\geq \ceil{\log(n_i^2 k^3)}\geq (2i+6)\log k=\log(k^{2i+6})
$$
where the second inequality is because of the definition of $n_{i-1}$. 
Hence, using Fact~\ref{fact:logupperbound} (whose assumptions $k\geq 3$ and $i\geq 2$ hold by the assumption of the theorem and claim respectively):
$$
n_i\geq k^{i+1}>(2i+8)\log k.
$$
Thus $n_i= \max \{(2i+8)\log k,\ceil{\log(n^2_{i+1}k^3)}\}$ must be equal to the second term in the $\max$, which concludes the proof.
\end{proof}

\noindent
Hence if $n_1=\ceil{\log(n^2_2k^3)}$, we can use the claim above to write 
$$
n_{i-1}= \ceil{2\log n_{i}+3\log k}\leq 4\log n_i, \quad \text{for } i\in \{2,\ldots ,r\},
$$
where the last inequality follows from $k\leq n^{1/3}_2\leq n^{1/3}_i$ (using $\ceil{\log(n^2_2k^3)}\geq 10\log k$ to conclude $k\leq n^{1/3}_2$ and Claim~\ref{claim:increasingdatabases}). Since $n_r=\log N$, it follows easily that $n_1=O(\log^{(r)}N)$. 

We conclude $n_1=O(\max\{\log k,\log^{(r)}N\})$.
\end{proof}

The following is our main result:

\begin{corollary}
\
\begin{itemize}
\item For every constant integer $r>0$ and sufficiently large $N=2^n$, there exist a quantum algorithm that finds a unique solution in a database of size $N$ with probability 1, using $(\frac{\pi}{4}+o(1))\sqrt{N}$ queries and $O(\sqrt{N}\log^{(r)}N)$ gates,
\item For every $\eps>0$ and sufficiently large $N=2^n$, there exist a quantum algorithm that finds a unique solution in a database of size $N$ with probability 1, using $\frac{\pi}{4}\sqrt{N}(1+\eps)$ queries and $O(\sqrt{N} \log (\log^{\star}N))$~gates.
\end{itemize}
\end{corollary}

\begin{proof}
Applying Corollary~\ref{cor:amplify1/kto1} to algorithm $\mathcal{C}^{(r)}$ (as described in Theorem~\ref{thm:grover-almost-opt}), with some $k\leq\log\log N$ to be specified later, we obtain an algorithm that succeeds with probability~1 using at most 
$$
\frac{\pi}{2}\Big(\sqrt{\frac{N}{4k}}\Big(1+\frac{4}{k}\Big)^{r}\Big)\cdot \Big(\sqrt{k}\Big(1+\frac{2}{\sqrt{k}}\Big)^2\Big) \leq \frac{\pi}{4} \sqrt{N}\Big(1+\frac{4}{\sqrt{k}}\Big)^{r+2}
$$
queries and 
$$
O\Big(\sqrt{k}n+\sqrt{N}\Big(1+\frac{6}{k}\Big)^{2r-1}\max\{\log k,\log^{(r)} N\}\Big) \leq O\Big(\sqrt{N}\Big(1+\frac{6}{k}\Big)^{2r} \max\{\log k,\log^{(r)} N\}\Big) 
$$ 
gates. To obtain the two claims of the corollary we can now either pick:
\begin{itemize}
\item $k=(c_1\log^{\star}N)^2$, where $c_1 \in [1,2]$ ensures $k$ is a power of~2. It follows that $(1+\frac{4}{c_1\log^{\star}N})^{r+2}=1+o(1)$. Since $\log^{\star}N\in o(\log^{(r)}N)$ for every constant $r$, we have $\max\{\log k,\log^{(r)} N\}=\log^{(r)} N$. Hence the query and gate complexities are $(\frac{\pi}{4}+o(1))\sqrt{N}$ and $O(\sqrt{N}\log^{(r)}N)$, respectively.
\item $r=\log^{\star} N$ and $k=(c_2(\log^{\star} N + 2))^2$, where we choose $c_2$ as the smallest number that is at least $4/\ln (1+\eps)$ and that makes $k$ a power of~2. We have $(1+\frac{4}{\sqrt{k}})^{r+2}\leq (1+\frac{4}{c_2(\log^{\star}N +2)})^{\log^{\star}N+2}\leq 1+\eps$. Hence the query and gate complexities are $\frac{\pi}{4}\sqrt{N}(1+\eps)$ and $O(\sqrt{N}\log (\log^{\star}N))$, respectively. 
\end{itemize}
\end{proof}

\section{Future work}
Our work could be improved further in a number of directions: 
\begin{itemize}
\item Can we remove the $\log (\log^{\star}) N$ factor in the gate complexity, reducing this to the optimal $O(\sqrt{N})$?  This may well be possible, but requires a different idea than our roughly $\log^{\star}$ recursion steps, which will inevitably end up with $\omega(\sqrt{N})$ gates.
\item Our construction only works for specific values of~$N$.  Can we generalize it to work for all sufficiently large $N$, even those that are not powers of~2, while still using close to the optimal $\frac{\pi}{4}\sqrt{N}$ queries?
\item Can we obtain a similar gate-optimized construction when the database has \emph{multiple} solutions instead of one unique one? Say when the exact number of solutions is known in advance? 
\item Most applications of Grover deal with databases with an unknown number of solutions, focus only on number of queries.  Are there application where our reduction in the number of elementary gates for search with one unique solution is both applicable and significant? 
\end{itemize}

\paragraph{Acknowledgments.}
We thank Peter H\o yer and Andris Ambainis for helpful comments related to~\cite{aaronson&ambainis:searchj}.

\bibliographystyle{alpha}

\newcommand{\etalchar}[1]{$^{#1}$}

\appendix

\section{Exact amplitude amplification}
\label{app:exactampproof}

For the sake of completeness we present the construction of quantum algorithm $\mathcal{B}$ from Theorem~\ref{thm:exactamplitude}. The idea is to lower the success probability from $a$ in such a way that an integer number of rounds of amplitude amplification suffice to produce a solution with probability exactly~$a'$.

Define $\theta=\frac{\arcsin(\sqrt{a'})}{2w+1}$ and $\tilde{a}=\sin^2(\theta)$, where $w$ is defined in Theorem~\ref{thm:exactamplitude}. Let $R_{\tilde{a}/a}$ be the one-qubit rotation that maps $\ket{0}\mapsto \sqrt{\tilde{a}/a}\ket{0}+\sqrt{1-\tilde{a}/a}\ket{1}$. Call an $(n+1)$-bit string $i,b$ a ``solution'' if $x_i=1$ \emph{and} $b=0$.  
Define the $(n+1)$-qubit unitary $O'_x=(I\otimes XH)O_x(I\otimes HX)$. It is easy to verify that $O'_x$ puts a $-$ in front of the solutions (in the new sense of the word), and a $+$ in front of the non-solutions.

Let $\mathcal{A}'=\mathcal{A}\otimes R_{\tilde{a}/a}$, and define $\ket{U}=\mathcal{A}'\ket{0^{n+1}}$ to be the final state of this new algorithm.
Let $\ket{G}$ be the normalized projection of $\ket{U}$ on the (new) solutions and $\ket{B}$ be the normalized projection of $\ket{U}$ on the (new) non-solutions. Measuring $\ket{U}$ results in a (new) solution with probability exactly~$\sin^2(\theta)$, hence we can write
$$
\ket{U}=\sin(\theta)\ket{G}+\cos(\theta)\ket{B}.
$$ 
Define $\mathcal{Q}=\mathcal{A}' D_{n+1}(\mathcal{A}')^{-1}O'_x$. This is a product of two reflections in the plane spanned by $\ket{G}$ and $\ket{B}$: $O'_x$ is a reflection through $\ket{G}$, and $\mathcal{A}'D_{n+1}(\mathcal{A}')^{-1}=2\ketbra{U}{U}-I$ is a reflection through~$\ket{U}$. As is well known in the analysis of Grover's algorithm and amplitude amplification, the product of these two reflections rotates the state over an angle $2\theta$. Hence after applying $\mathcal{Q}$ $w$ times to $\ket{U}$ we have the state
$$
\mathcal{Q}^w\ket{U}=\sin((2w+1)\theta)\ket{G}+\cos((2w+1)\theta)\ket{B}=\sqrt{a'}\ket{G}+\sqrt{1-a'}\ket{B},
$$
since $(2w+1)\theta=\arcsin(\sqrt{a'})$. Thus the algorithm $\mathcal{A}'$ can be boosted to success probability~$a'$ using an integer number of applications of $\mathcal{Q}$.

Our new algorithm $\mathcal{B}$ is now defined as $\mathcal{Q}^w\mathcal{A}'$. It acts on $n+1$ qubits (all initially~0) and maps 
$$
\ket{0^{n+1}}\mapsto\sqrt{a'}\ket{G}+\sqrt{1-a'}\ket{B},
$$ 
so it finds a solution with probability exactly~$a'$. 
$\mathcal{B}$ uses $w+1$ applications of algorithm $\mathcal{A}$ together with elementary gate~$R_{\tilde{a}/a}$; $w$ applications of $\mathcal{A}^{-1}$ together with~$R_{\tilde{a}/a}^{-1}$; $w$ applications of $O'_x$ (each of which involves one query to~$x$ and two other elementary gates, counting $XH$ as one gate); and $w$ applications of $D_{n+1}$ (each of which takes $4n+3$ elementary gates). Hence the total number of queries that $\mathcal{B}$ makes is at most $(2w+1)Q+w$ and the number of gates used by $\mathcal{B}$ is at most $(2w+1)E+4w(n+2)$.

\end{document}